\documentclass[letterpaper,USenglish]{lipics}
\usepackage{float}

\title{Subword Complexity and $k$-Synchronization}
\titlerunning{Subword Complexity} %optional, in case that the title is too long; the running title should fit into the top page column

\author{Daniel Go\v{c}, Luke Schaeffer, and Jeffrey Shallit}
\affil{School of Computer Science, University of
Waterloo, Waterloo, ON  N2L 3G1 Canada\\
\texttt{\{dgoc,l3schaeffer,shallit\}@uwaterloo.ca}}
\authorrunning{D. Go\v{c}, L. Schaeffer, and J. Shallit} %optional. First: Use abbreviated first/middle names. Second (only in severe cases): Use first author plus 'et. al.'

% mandatory: Please provide 1-5 keywords
\def\Enn{\mathbb{N}}

\begin{document}

\maketitle

\begin{abstract}
We show that the subword complexity function $\rho_{\bf x} (n)$, which
counts the number of distinct factors of length $n$  of a sequence
$\bf x$, is $k$-synchronized in the sense of Carpi if $\bf x$ is
$k$-automatic.  As an
application, we generalize recent results of Goldstein.  We give
analogous results for the number of distinct factors of length $n$ that
are primitive words or powers.  In contrast, we show that the function
that counts the number of unbordered factors of length $n$ is {\it
not\/} necessarily $k$-synchronized for $k$-automatic sequences.
\end{abstract}

\section{Introduction}

We are concerned with the representation of integers in base $k$, where
$k \geq 2$ is an integer.  We let $\Sigma_k = \lbrace
0,1,2, \ldots, k-1 \rbrace$, and we let
$(n)_k$ denote the canonical
representation of $n$ in base $k$, starting with the most
significant digit, and {\it without\/} leading zeroes.
If $x \in \Sigma_k^*$, we let $[x]_k$ denote the integer
represented by $x$ (where $x$ is allowed to have leading zeroes).
To represent a pair of integers $(m, n)$, we use words over the
alphabet $\Sigma_k \times \Sigma_k$.  For such a word $x$,
we let $\pi_i (x)$
to be the projection onto the $i$'th coordinate.  The canonical
representation $(m,n)_k$ is defined to be the word $x$ such
that $[\pi_1(x)]_k = m$ and $[\pi_2(x)]_k = n$, and having
no leading $[0,0]$'s.  For example
$(43,17)_2 = [1,0][0,1][1,0][0,0][1,0][1,1]$.  

Recently, Arturo Carpi and his co-authors
\cite{Carpi&Maggi:2001,Carpi&DAlonzo:2009,Carpi&DAlonzo:2010}
introduced a very interesting class of sequences that are computable by
automata in a novel fashion:  the class of $k$-synchronized sequences.
Let $(f(n))_{n \geq 0}$ be a sequence taking values in $\Enn$.  They
call such a sequence {\it $k$-synchronized\/} if there is a
deterministic finite automaton $M$ accepting the base-$k$
representation of the graph of $f$, namely $\{ (n, f(n))_k  \ : \ n
\geq 0 \}$.  

Sequences that are $k$-synchronized are ``halfway between'' the class
of $k$-automatic sequences, introduced by Cobham \cite{Cobham:1972} and
studied in many papers; and the class of $k$-regular sequences,
introduced by Allouche and Shallit
\cite{Allouche&Shallit:1992,Allouche&Shallit:2003}.  They are
particularly interesting for two reasons.  If a sequence $(f(n))$ is
$k$-synchronized, then
\begin{itemize}
\itemindent 10pt
\item[(a)] we immediately get a bound on its growth rate:  
$f(n) = O(n)$;

\item[(b)] we immediately get a linear-time algorithm for efficiently
calculating $f(n)$.
\end{itemize}

Result (a) can be found in \cite[Prop.\ 2.5]{Carpi&Maggi:2001}.
We now state and prove result (b).

\begin{theorem}
Suppose $(f(n))_{n \geq 0}$ is $k$-synchronized.  Then there is an algorithm
that, given the base-$k$ representation of $n$, will compute 
the base-$k$ representation of $f(n)$ in $O(\log n)$ time.
\label{calc}
\end{theorem}

\begin{proof}
We know there is a 
DFA $M = (Q, \Sigma_k \times \Sigma_k, \delta,
q_0, F)$ accepting
$L = \lbrace (n, f(n))_k \ : \ n \geq 0 \rbrace$.
From result (a) above
we know that $f(n) \leq Cn $, for some constant $C$, so if
$(n, f(n))_k$ is accepted, then the first component is
$0^s w$ for some $s \leq \log_k C$, where
$w$ is the canonical base-$k$ representation of $n$.  
Let $N = s + |w|$.
We now create a directed graph out of $N+1$ copies of
the transition graph for $M$, by starting at the final states
of $M$ and tracing a path backwards,
using the reversed transitions of $M$.  This path is chosen so the
first component of the labels encountered
form $w^R 0^s$ and the second component
is arbitrary.  The reslting graph has at most $O(N)$ transitions and vertices.
There will be only one path of length $l$ with $|w| \leq l \leq N$
that leads to the initial state $q_0$, and this can be found with
depth-first search in $O(N)$ time.  Then, reading the corresponding
labels of the second components in the forward direction gives the
base-$k$ representation of $f(n)$.
\end{proof}

In this paper,
we are concerned with infinite words over a {\it finite\/} alphabet.
Let ${\bf x} = a_0 a_1 a_2 \cdots$ be an infinite word.
By ${\bf x}[m..n]$ we mean the factor $a_m a_{m+1} \cdots a_n$ of
$\bf x$ of length $n-m+1$.
The subword complexity function $\rho_{\bf x} (n)$ counts the number
of distinct factors of length $n$.  

An infinite word or sequence
${\bf x}$ is said to be $k$-automatic if there is an automaton with
outputs associated with the states that, on input $(n)_k$, reaches a state
with output ${\bf x}[n]$.
In this paper we show that if $\bf x$ is a $k$-automatic sequence, then
the subword complexity $\rho_{\bf x} (n)$ is $k$-synchronized. 
As an application, we generalize and simplify recent results of
Goldstein \cite{Goldstein:2009,Goldstein:2011}.
Furthermore, we obtain
analogous results for the number of length-$n$ primitive words and
the number of length-$n$ powers.  

We remark that there are a number of quantities about $k$-automatic
sequences already known to be
$k$-synchronized.  These include
\begin{itemize}
\item the separator sequence of a non-ultimately-periodic $k$-automatic sequence
	\cite{Carpi&Maggi:2001};
\item the repetitivity index of a $k$-automatic sequence
	\cite{Carpi&DAlonzo:2009};
\item the recurrence function of a $k$-automatic sequence
	\cite{Charlier&Rampersad&Shallit:2011};
\item the ``appearance'' function of a $k$-automatic sequence 
	\cite{Charlier&Rampersad&Shallit:2011}.
\end{itemize}
The latter two examples were not explicitly stated to be $k$-synchronized
in \cite{Charlier&Rampersad&Shallit:2011}, but the result follows immediately
from the proofs in that paper.

\section{Subword complexity}

Cobham \cite{Cobham:1972} proved that if $\bf x$ is a $k$-automatic
sequence, then $\rho_{\bf x} (n) = O(n)$.  Cassaigne \cite{Cassaigne:1996}
proved that
any infinite word $\bf x$ satisfying $\rho_{\bf x} (n) = O(n)$ also
satisfies $\rho_{\bf x} (n+1) - \rho_{\bf x} (n) = O(1)$.
Carpi and D'Alonzo \cite{Carpi&DAlonzo:2010} 
showed that
the subword complexity function
$\rho_{\bf x} (n)$ is a $k$-regular sequence.

Charlier, Rampersad, and Shallit \cite{Charlier&Rampersad&Shallit:2011}
found this result independently, using a somewhat different approach.
They used the following idea.  Call an occurrence of the
factor $t = {\bf x}[i..i+n-1]$
``novel'' if $t$ does not appear as a factor of ${\bf x}[0..i+n-2]$.  
In other words, the leftmost occurrence of $t$ in $\bf x$
is at position $i$.  Then the number of factors of length $n$ in
$\bf x$ is equal to the number of novel occurrences of factors of length $n$.
The property that ${\bf x}[i..i+n-1]$ is novel
can be expressed as a predicate, as follows:
\begin{multline}
\{ (n,i)_k \ : \ \forall j, 0 \leq j < i \ 
	{\bf x}[i..i+n-1] \not= {\bf x}[j..j+n-1] \} 
= \\
\{ (n,i)_k \ : \ \forall j, 0 \leq j < i  \ 
	\exists m, 0 \leq m < n \  
		{\bf x}[i+m] \not= {\bf x}[j+m] \} .
\label{pred}
\end{multline}

As shown in \cite{Charlier&Rampersad&Shallit:2011}, the base-$k$
representation of the integers satisfying
any predicate
of this form (expressible using quantifiers, integer addition and
subtraction, indexing into a $k$-automatic sequence
$\bf x$, logical operations, and comparisons)
can be accepted by an explicitly-constructable
deterministic finite automaton.  From this, it follows that
the sequence $\rho_{\bf x} (n)$ is $k$-regular, and hence
can be computed explicitly in terms of the product of certain
matrices and vectors depending on the base-$k$ expansion of $n$.

We show that, in fact, the subword complexity function
$\rho_{\bf x}(n)$ is $k$-synchronized.  
The main observation needed is the following (Theorem~\ref{luke}):
in any sequence of linear
complexity, the novel occurrences of 
factors are ``clumped together'' in a bounded
number of contiguous blocks.  This makes it easy to count them.

More precisely, let $\bf x$ be an infinite word and
for any $n$ consider the set of novel occurrences
$$ E_{\bf x}(n) := \lbrace i \ : \text{ the occurrence }
{\bf x}[i..i+n-1] \text{ is novel } \rbrace.$$
We consider how $E_{\bf x}(n)$ evolves with increasing $n$.

As an example, consider the Thue-Morse sequence 
$${\bf t} = t_0 t_1 t_2 \cdots = 0110100110010110 \cdots,$$
defined by letting $t_n$ be the number of $1$'s in the binary expansion of
$n$, taken modulo $2$.  The gray squares in
the rows of
of Figure~\ref{fig1} depict the members of $E_{\bf t}(n)$ for the
Thue-Morse sequence for $1 \leq n \leq 9$.

\begin{figure}[H]
\begin{center}
\resizebox{12cm}{!}{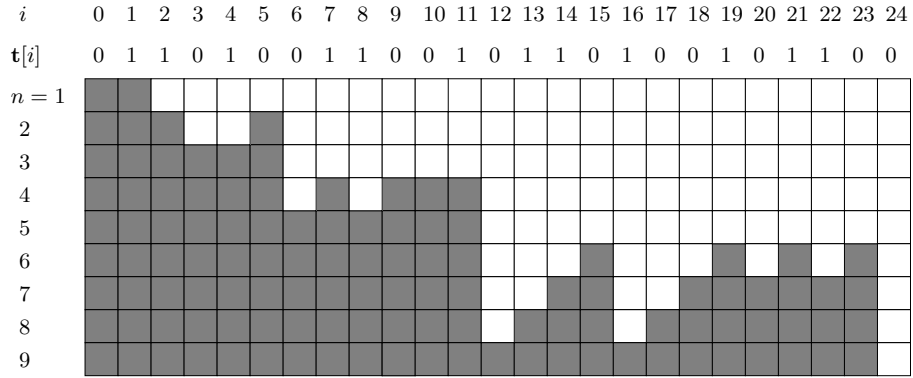}
\end{center}
\caption{Evolution of novel occurrences of factors in the Thue-Morse sequence}
\label{fig1}
\end{figure}

\begin{lemma}
Let $\bf x$ be an infinite word.
If the factor of length $n$ beginning at position $i$ is a novel occurrence,
so is
\begin{itemize}
\itemindent 10pt
\item[(a)] the factor of length $n+1$ beginning at position $i$;
\item[(b)] the factor of length $n+1$ beginning at position $i-1$
(for $i \geq 1$).
\end{itemize}
\end{lemma}

\begin{proof}
(a) Suppose the factor of length $n+1$ also occurs at some position
$j < i$.  Then the factor of length $n$ also occurs at position $j$,
contradicting the fact that it was a novel occurrence at $i$.

(b) Suppose the factor of length $n+1$ beginning at position $i-1$
occurs at some earlier position $j < i-1$.    We can write the
factor as $ax$, where $a$ is a single letter and $x$ is a word,
so the factor of length $n$ beginning at position $i$ must also
occur at position $j+1 < i$.  But then it is not a novel occurrence.
\end{proof}

\begin{theorem}
Let $\bf x$ be an infinite word.
For $n \geq 1$,
the number of contiguous blocks in $E_{\bf x}(n)$  is at most
$\rho_{\bf x} (n) - \rho_{\bf x} (n-1) + 1$.
\label{luke}
\end{theorem}

\begin{proof}
We prove the claim by induction on $n$.  For $n = 1$ the claim says
there are at most $\rho_{\bf x}(1)$ contiguous blocks, which is evidently
true, since there are at most $\rho_{\bf x}(1)$ novel factors of length $1$.

Now assume the claim is true for all $n' < n$; we prove it for $n$.
Consider the evolution of the novel occurrences of factors in going from length
$n-1$ to $n$.
Every occurrence that was previously novel is still novel, and furthermore
in every contiguous block except the first, we get novel occurrences
at one position to the left of the beginning of the block.  
So if row $n-1$ has $t$ contiguous blocks, then we get $t-1$ novel 
occurrences
at the beginning of each block, except the first.  (Of course, the
first block begins at position $0$, since any factor beginning at
position $0$ is novel, no matter what the length is.)
The remaining
$ \rho(n) - \rho(n-1) - (t-1)$ novel occurrences could be, in the worst case,
in their own individual contiguous blocks. Thus
row $n$ has at most $ t + \rho(n) - \rho(n-1) - (t-1) = \rho(n)+\rho(n-1) +1$
contiguous blocks.  
\end{proof}

In our Thue-Morse example, it is well-known that 
$\rho_{\bf t}(n) -\rho_{\bf t}(n-1) \leq 4$,
so the number of contiguous blocks in any row is at most $5$.  This
is achieved, for example, for $n = 6$.

\begin{example}
We give an example of a recurrent infinite word over a finite alphabet
where the number of contiguous blocks in $E_{\bf x}(n)$ is unbounded.
Consider the word
\begin{displaymath}
{\bf w} =  \prod_{n \geq 1}  (n)_2 
= 1 10 11 100 101 110 111 1000 \cdots .
\end{displaymath}
Then for each $n \geq 5$
the first occurrence of each of the words
$0^{n-1} 1$, $ 0^{n-2} 11, \ldots, 0^2 1^{n-2}$
have a non-novel occurrence immediately following them, which 
shows there at at least $n-2$ blocks in
$E_{\bf w} (n)$.
\end{example}

\begin{corollary}
If $\rho_{\bf x} (n) = O(n)$, then there is a constant $C$ such that
every row $E_{\bf x}(n)$ in the evolution of
novel occurrences consists of at most $C$ contiguous blocks.
\end{corollary}

\begin{proof}
By the result of Cassaigne \cite{Cassaigne:1996}, we know that there exists
a constant $C$ such that $\rho_{\bf x} (n)- \rho_{\bf x} (n-1) \leq C-1$.
By Theorem~\ref{luke}, we know there are at most $C$ contiguous blocks
in any $E_{\bf x}(n)$.
\end{proof}

\begin{theorem}
Let $\bf x$ be a $k$-automatic sequence.  Then its
subword complexity function $\rho_{\bf x} (n)$ is $k$-synchronized.
\label{scthm}
\end{theorem}

\begin{proof}
Following \cite{Charlier&Rampersad&Shallit:2011}, 
it suffices to show how to accept the
language
$$\{ (n, m)_k \ : \ n \geq 0 \text{ and }  m = \rho_{\bf x} (n) \rbrace$$
with a finite automaton.  Here is a sketch of the argument.
From our results above, we know that there
is a finite constant $C \geq 1$ such that the number of contiguous
blocks in any row of the factor evolution diagram is bounded by $C$.
So we simply ``guess'' the endpoints of every block and then verify
that each factor of length $n$ starting at the positions inside blocks
is a novel occurrence, while all other factors are not.  Finally, we verify that
$m$ is the sum of the sizes of the blocks.

To fill in the details, we observe above in (\ref{pred}) that the predicate
``the factor of length $n$ beginning at position $i$ of $\bf x$ is a
novel occurrence'' is solvable by a finite automaton.  Similarly,
given endpoints $a, b$ and $n$, the predicates ``every factor of length $n$ 
beginning at positions $a$ through $b$ is a novel occurrence''
and ``no factor of length $n$ beginning at positions $a$ through $b$
is a novel occurrence'' are also solvable by a finite automaton.
The length of each block is just $b-a+1$, and it is easy to create
an automaton that will check if the sums of the lengths of the blocks
equals $m$, which is supposed to be $\rho_{\bf x} (n)$.
\end{proof}

Applying Theorem~\ref{calc} we get

\begin{corollary}
Given a $k$-automatic sequence $\bf x$, there is an algorithm that, on input
$n$ in base $k$, will produce $\rho_{\bf x} (n)$ in base $k$ in time 
$O(\log n)$. 
\end{corollary}

As another application, we can recover and improve some recent results of
Goldstein \cite{Goldstein:2009,Goldstein:2011}.  He showed how to 
compute the quantities $\limsup_{n \geq 1} \rho_{\bf x} (n)/n$ and
$\liminf_{n \geq 1} \rho_{\bf x} (n)/n$  for the special case
of $k$-automatic sequences that are the fixed points of $k$-uniform
morphisms related to certain groups.
Corollary~\ref{ss} below generalizes
these results to all $k$-automatic sequences.

\begin{corollary}
There is an algorithm, that,
given a $k$-automatic sequence $\bf x$, will compute
$\sup_{n \geq 1} \rho_{\bf x} (n)/n$, 
$\limsup_{n \geq 1} \rho_{\bf x} (n)/n$,
and $\inf_{n \geq 1} \rho_{\bf x} (n)/n$,
$\liminf_{n \geq 1} \rho_{\bf x} (n)/n$.
\label{ss}
\end{corollary}

\begin{proof}
We already showed how to construct an automaton accepting
$\lbrace (n, \rho_{\bf x} (n))_k \ : \ n \geq 1 \rbrace$.
Now we just use the results from \cite{Shallit:2011,Schaeffer&Shallit:2012}.
Notice that the
$\limsup$ corresponds to what is called the largest ``special point'' in
\cite{Schaeffer&Shallit:2012}.
\end{proof}  

\begin{example}
Continuing our example of the Thue-Morse sequence,
Figure~\ref{fig2} displays a DFA accepting
$$ \{  (n, \rho_{\bf t} (n))_k  \ : \ n \geq 0 \} .$$
Inputs are given with the most significant digit first; the ``dead'' state
and transitions leading to it are omitted.
\end{example}

Given an infinite word $\bf x$,
we can also count the number of contiguous blocks in each $E_{\bf x}(n)$ for $n \geq 0$.
(For the Thue-Morse sequence this gives the sequence
$1,1,2,1,3,1,5,3,3,1, \ldots$.)  If $\bf x$ is $k$-automatic, then
this sequence is also, as the following theorem shows:

\begin{theorem}
If $\bf x$ is $k$-automatic then the sequence $(e(n))_{n \geq 0}$ counting
the number of contiguous blocks in the $n$'th step $E_{\bf x}(n)$ of the
evolution of novel occurrences of factors in $\bf x$ is also $k$-automatic.
\end{theorem}

\begin{proof}
Since we have already shown that the number of contiguous blocks is bounded
by some constant $C$
if $\bf x$ is $k$-automatic, it suffices to show for each $i \leq C$
we can create an automaton
to accept the language
$\{ (n)_k \ : \ E_{\bf x}(n) \text{ has exactly $i$ contiguous blocks } \}$.
To do so, on input $n$ in base $k$ we guess the endpoints of the
$i$ contiguous nonempty blocks, verify that the length-$n$ occurrences
at those positions are novel, and that all other occurrences are not novel.
\end{proof}

\begin{figure}[H]
\centering
\includegraphics[scale=0.3]{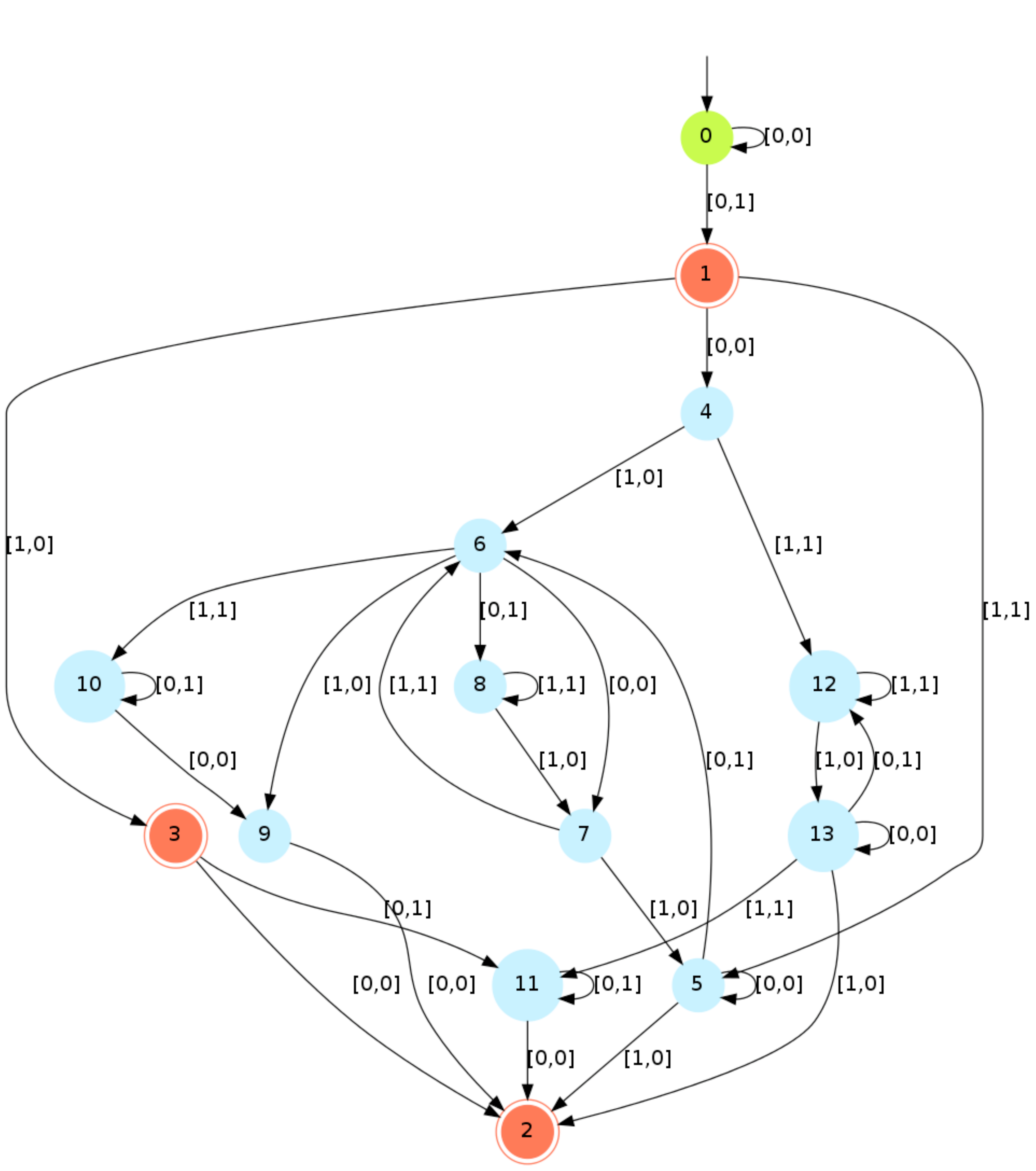}
\caption{Automaton computing the subword complexity of the Thue-Morse sequence}
\protect\label{fig2}
\end{figure}

\begin{example}
Figure~\ref{fig3} below gives the automaton computing the number $e(n)$
of contiguous blocks of novel occurrences of length-$n$ factors
for the Thue-Morse sequence.
Here is a brief table:
\begin{table}[H]
\begin{center}
\begin{tabular}{|c|r|r|r|r|r|r|r|r|r|r|r|r|r|r|r|}
\hline
$n$ & 0 & 1 & 2 & 3 & 4 & 5 & 6 & 7 & 8 & 9 & 10 & 11 & 12 & 13 & 14  \\
\hline
$e(n)$ & 1 & 1 & 2 & 1 & 3 & 1 & 5 & 3 & 3 & 1 & 5 & 5 & 5 & 3 & 3 \\
\hline
\end{tabular}
\end{center}
\end{table}
\end{example}

\section{Implementation}

We wrote a program that, given an automaton generating
a $k$-automatic sequence $\bf x$,
will produce a deterministic finite automaton
accepting the language $ \{ (n, \rho_{\bf x} (n))_k \ : \ n \geq 0 \}$.
We used the following variant which does not require advance
knowledge of the bound on the first difference of $\rho_{\bf x} (n)$:
\begin{enumerate}
\item Construct an automaton $R$ that accepts $(n,s,e,\ell)$ if,
for factors of length $n$,
the next contiguous block of novel occurrences
after position $s$ ends at position $e$ and has length
$\ell$. If there are no blocks past $s$, accept $(n,s,s,0)$.

\item Construct an automaton $M_0$ that accepts $(n, 0, 0)$.

\item Construct an automaton $M_{j+1}$ that accepts $(n, S, e)$
if there exist $s$ and $S'$ such that
\begin{itemize}
\itemindent 10pt
\item[(i)] $M_j$ accepts $(n, S', s)$ 
\item[(ii)] $R$ accepts $(n, s, e, S - S')$.
\end{itemize}
\item If $M_{j+1} = M_j$ then we are done. We create an automaton that accepts 
$(n, S)$ if there exists $e$ such that $M_j$ accepts $(n, S, e)$.
\end{enumerate}

\begin{figure}[H]
\centering
\includegraphics[width=10cm]{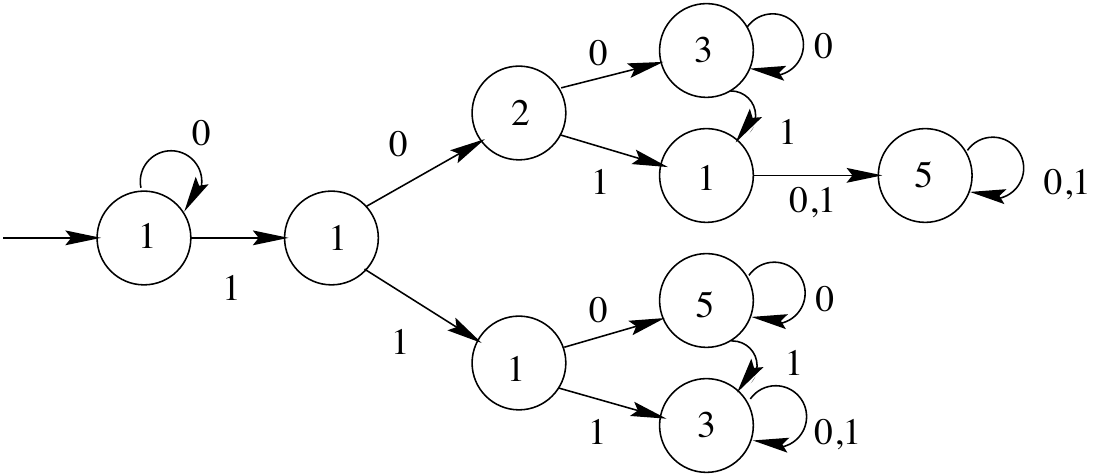}
\caption{Automaton computing number of contiguous blocks of novel
occurrences of length-$n$ factors in the Thue-Morse sequence}
\protect\label{fig3}
\end{figure}

Besides the automaton depicted in Figure~\ref{fig1}, we ran our program
on the paperfolding sequence \cite{Dekking&MendesFrance&vanderPoorten:1982}
and the so-called ``period-doubling sequence'' \cite{Damanik:2000}.
The results are depicted below in Figures~\ref{fig4} and \ref{fig5}.

\section{Powers and primitive words}

Let $w$ be a nonempty word.  We say $w$ is a {\it power} if there
exists a word $x$ and an integer $k \geq 2$ such that $w = x^k$;
otherwise we say $w$ is {\it primitive}.
Given a word $z$, there is a unique way to write it as $y^i$, where
$y$ is primitive and $i$ is an integer $\geq 1$; this $y$ is called
the {\it primitive root} of $z$.  Thus, for example, the primitive
root of {\tt murmur} is {\tt mur}.

We say $w = a_1 \cdots a_n$ has a {\it period\/} $p$ if $a_i = a_{i+p}$
for $1 \leq i \leq n-p$.  Thus, for example, {\tt alfalfa} has period $3$.
It is easy to see that a word $w$ is a power
if and only if it has a period $p$ such that $p < |w|$ and $p \ | \ |w|$.

Two finite words $x, y$ are conjugates if one is a cyclic shift of
the other; in other words, if there exist words $u, v$ such that
$x = uv$ and $y = vu$.  For example, {\tt enlist} is a conjugate of
{\tt listen}.  As is well-known, every conjugate of a power of 
a word $x$ is a power of a conjugate of $x$.  The lexicographically
least conjugate of a primitive word is called a {\it Lyndon word}.
We call the lexicographically least conjugate of the primitive root of
$x$ the {\it Lyndon root} of $x$.

The following lemma says that if we consider the starting positions of
length-$n$ powers in a word $x$, then there must be large gaps
between contiguous blocks of such starting positions.

\begin{lemma}
Let $z$ be a finite or infinite word, and let $n \geq 2$ be an
integer.  Suppose there exist integers $i, j$ such that
\begin{itemize}
\itemindent 10pt
\item[(a)] $w_1 := z[i..i+n-1]$ is a power;
\item[(b)] $w_2 := z[j..j+n-1]$ is a power;
\item[(c)] $i < j \leq i + n/3$.
\end{itemize}
Then $z[t..t-n-1]$ is a power for $i \leq t \leq j$.
Furthermore, if $x_1$ is the Lyndon root of $w_1$, then
$x_1$ is also the Lyndon root of each word $z[t..t-n-1]$.
\label{gocl}
\end{lemma}

\begin{figure}[H]
\centering
\includegraphics[width=7cm]{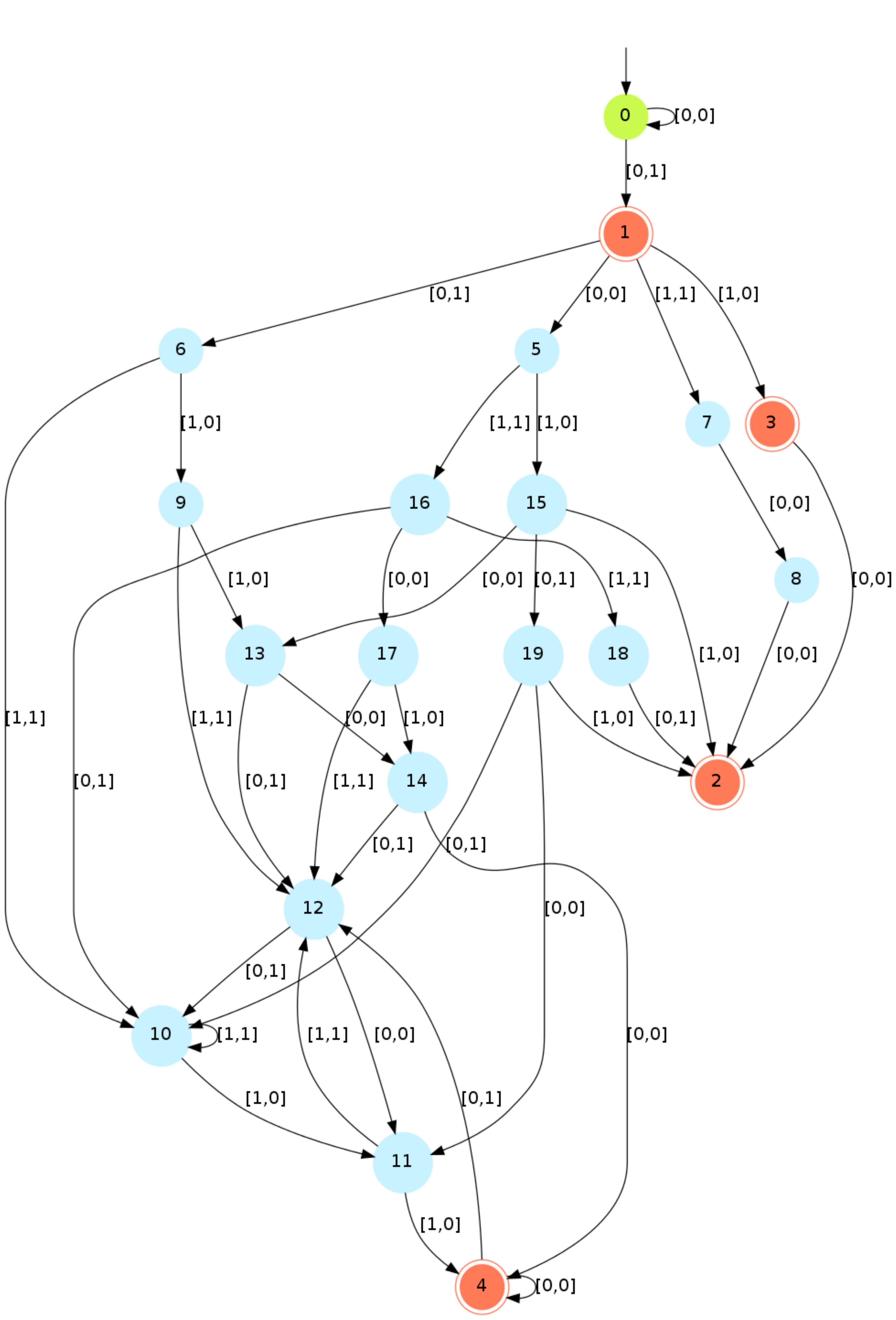}
\caption{Automaton computing the subword complexity of the paperfolding sequence}
\protect\label{fig4}
\end{figure}

\begin{figure}[H]\centering
\includegraphics[width=7cm]{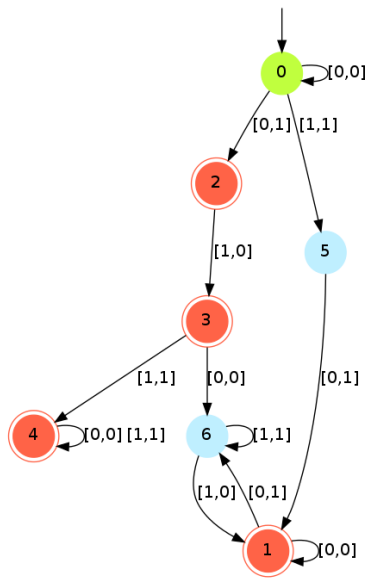}
\caption{Automaton computing
the subword complexity of the period-doubling sequence}
\protect\label{fig5}
\end{figure}

\begin{proof}
Let $x_1$ be the primitive root of $w_1$ and $x_2$ be the 
primitive root of $w_2$.  Since $x_1$ and $x_2$ are powers, there
exist integers $p_1, p_2 \geq 2$ such that $w_1 = x_1^{p_1}$
and $w_2 = x_2^{p_2}$.

Since $w_1$ and $w_2$ are both of length $n$, and since their starting
positions are related by $i < j \leq i + n/3$, it follows that the word
$v := z[j..i+n-1]$ is common to both $w_1$ and $w_2$, and $|v| = i+n-j
\geq i + 2n/3 + n/3 - j \geq 2n/3$.

Now there are three cases to consider:
\begin{itemize}
\itemindent 10pt
\item[(a)] $|x_1|>|x_2|$;
\item[(b)] $|x_1|<|x_2|$;
\item[(c)] $|x_1|=|x_2|$.
\end{itemize}

Case (a):  We must have $p_2 > p_1 \geq 2$, so $p_2 \geq 3$.
Since $v$ is a suffix of $w_1$, it has period
$|x_1| \leq n/2$.  Since $v$ is a prefix of $w_2$, it
has period $|x_2| \leq n/3$.   Then $|v| \geq 2n/3 \geq |x_1| + |x_2|$.
By a theorem of Fine and Wilf \cite{Fine&Wilf:1965},
it now follows that $v$, and hence $x_1$,
has period $p := \gcd(|x_1|,|x_2|) \leq |x_2| < |x_1|$.
Now $p$ is less than $|x_1|$ and also divides it, so this means
$x_1$ is a power, a contradiction, since we assumed $x_1$ is primitive.
So this case cannot occur.

Case (b) gives a similar contradiction.

Case (c):  We have $p_1 = p_2 \geq 2$.    Then the last occurrence
of $x_1$ in $w_1$ lies inside $x_2^2$, and so $x_1$ is a conjugate
of $x_2$.  Hence $w_1$ is a conjugate of $w_2$.
It now follows that $z[t..t+n-1]$ is a conjugate
of $w_1$ for every $t$,
$i \leq t \leq j$.  But the conjugate of a power is itself
a power, and we are done.
\end{proof}

\medskip

\begin{remark}
The bound of $n/3$ in the statement of Lemma~\ref{gocl} is 
best possible, as shown by the following class of examples.
Let $h$ be the morphism that maps $1 \rightarrow 21$ and
$2 \rightarrow 22$, and consider the word
$$ h^i (122122121212).$$
This word is of length $12 \cdot 2^i$, and
contains squares of
length $3 \cdot 2^{i+1}$ starting in the first $3 \cdot 2^i$ 
positions, and cubes of length
$3 \cdot 2^{i+1}$ ending in the last $2^i + 1$ positions.
This achieves a gap of $n/3 + 1$ infinitely often.
\end{remark}

\medskip

Now, given an infinite word $\bf x$, we define a function
$\alpha_{\bf x} (n)$, the {\it appearance function}, to be the 
least index $i$ such that every length-$n$ factor of $\bf x$
appears in the prefix ${\bf x}[0..i+n-1]$; see 
\cite[\S 10.10]{Allouche&Shallit:2003b}.

\begin{theorem}
If $\bf x$ is a $k$-automatic sequence, then
$\alpha_{\bf x} (n) = O(n)$.
\label{app-thm}
\end{theorem}

\begin{proof}
First, we show that the appearance function is $k$-synchronized.
It suffices to show that there is an automaton accepting
$\lbrace (n,m)_k \ : \  m = \alpha_{\bf x} (n) \rbrace$.
To see this, note that on input $(n,m)_k$ we can
check that
\begin{itemize}
\item for all $i \geq 0$ there exists $j$, $0 \leq j \leq m$
such that ${\bf x}[i..i+n-1] = {\bf x}[j..j+n-1]$; and 
\item for all $l < m$ we have ${\bf x}[m..m+n-1] \not= {\bf x}[l..l+n-1]$.
\end{itemize}
From \cite[Prop.\ 2.5]{Carpi&Maggi:2001}
we know $k$-synchronized functions are $O(n)$.
\end{proof}

As before, we consider maximal blocks of novel occurrences of
length-$n$ powers in $\bf x$.  Our goal is to prove 

\begin{lemma}
If $\bf x$ is $k$-automatic, then
there are only a constant number of such blocks.
\label{lnpow}
\end{lemma}

\begin{proof}
To begin with, we consider maximal blocks of length-$n$ powers
in $\bf x$ (not considering whether they are novel occurrences).
From Theorem~\ref{app-thm} we know that every length-$n$ factor
must occur at a position $< Cn$, for some constant $C$
(depending on $\bf x$).  We first argue that the number of
maximal blocks of length-$n$ powers, up to the position of the
last length-$n$ power to occur for the first time,
is at most $3C$.    

Suppose there $\geq 3C+1$ such blocks.  Then Lemma~\ref{gocl}
says that any two such blocks must be separated by at least
$n/3$ positions.  So the first occurrence of the last factor
to occur occurs at a position $\geq (3C)(n/3) = Cn$,
a contradiction.  

So using a constant number of blocks, in which each position
of each block starts a length-$n$ factor that is a power,
we cover the starting positions of all such factors.  It now
remains to process these blocks to remove occurrences of 
length-$n$ powers that are not novel.

The first thing we do is remove from each block the positions
starting length-$n$ factors that have already occurred 
{\it in that block}.  This has the effect of truncating long
blocks.  The new blocks have the property that each factor
occurring at the starting positions in the blocks never appeared
before in that block.

Above we already proved that inside each block, the powers that
begin at each position
are all powers of some conjugate of a fixed Lyndon word.
Now we process the blocks associated with the same Lyndon root
together, from the first (leftmost) to the last.
At each step, we remove from the current block all the positions
where length-$n$ factors begin that have appeared in any previous block.
When all blocks have been processed, we need to see that there are
still at most a constant number of contiguous blocks remaining.

Suppose the associated Lyndon root is $y$, with $|y| = d$.  Each position
in a block is the starting position of a power of a conjugate of $y$, and
hence corresponds to 
a right rotation of $y$ by some integer  $i$, $0 \leq i < d$.  Thus
each block $B_j$ actually corresponds to some $I_j$ that is a contiguous
subblock of $0, 1, \ldots, d-1$
(thought of as arranged in a circle).  

As we process the blocks associated with $y$ from left to right 
we replace $I_j$ with $I'_j := I_j - (I_1 \ \cup \ \cdots \ \cup I_{j-1})$.
Now if $I \subseteq \lbrace 0, 1,\ldots, d-1 \rbrace$ is a union 
of contiguous subblocks, let $\#I$ be the number of contiguous subblocks
making up $I$.  We claim that
\begin{equation}
\#I'_1 + \#I'_2 + \cdots + \#I'_n + \#(\bigcup_{1 \leq i \leq n} I'_i) \leq 2n
.
\label{ineq}
\end{equation}

To see this, suppose that when we set $I'_n := I_n - (I_1 \ \cup \  \cdots \ \cup
I_{n-1})$, the subblock $I_n$ has an intersection with $t$ of
the lower-numbered subblocks.   Forming the union
$(\bigcup_{1 \leq i \leq n} I'_i)$ then obliterates $t$ subblocks and
replaces them with $1$.  But $I'_n$ has $t-1$ new subblocks, plus at most $2$
at either edge (see Figure~\ref{fig9}).  This means that the left side of
(\ref{ineq}) increases by at most $(1-t) + (t-1) + 2 = 2$.  Doing this
$n$ times gives the result.

\begin{figure}[H]
\begin{center}
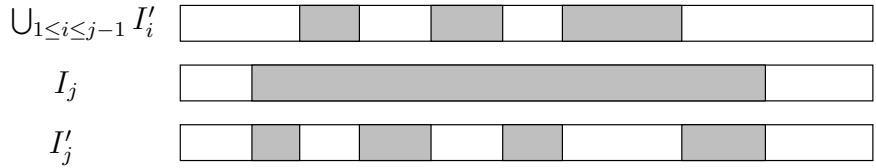
\end{center}
\caption{How the number of blocks changes}
\label{fig9}
\end{figure}

Now at the end of the procedure
there will be at least one interval in the union of all the $I_i$,
so $\#I'_1 + \#I'_2 + \cdots + \#I'_n \leq 2n-1$, and we have
proved (\ref{ineq}).

Earlier we showed that there are at most $3C$ maximal blocks of
length-$n$ powers, up to the position of the last length-$n$ power
to occur for the first time.  Then, after processing these blocks
to remove positions corresponding to factors that occurred earlier,
we will have at most $2(3C) = 6C$ blocks remaining.
\end{proof}

\begin{corollary}
If $\bf x$ is $k$-automatic, then 
\begin{itemize}
\item the function counting the number of distinct length-$n$ factors
that are powers is $k$-synchronized;
\item the function counting the number of distinct length-$n$ factors
that are primitive words is $k$-synchronized.
\end{itemize}
\end{corollary}

\begin{proof}
Suppose $\bf x$ is $k$-automatic, and generated by the DFAO $M$.
From the Lyndon-Sch\"utzenberger theorem \cite{Lyndon},
we know that
a word $x$ is a power if and only if there exist nonempty words
$y, z$ such that $x = yz = zy$.  
Thus, we can
express the predicate $P(i,j) :=$ ``${\bf x}[i..j]$ is a power'' as
follows:  ``there does not exist $d$,
$0 < d < j-i+1 $, such that ${\bf x}[i..j-d] = {\bf x}[i+d..j]$ and
${\bf x}[j-d+1..j] = {\bf x}[i..i+d-1]$''.    Furthermore, we can
express the predicate $P'(i,n) :=$ ``${\bf x}[i..i+n-1]$ is a length-$n$
power and the first occurrence of that power in ${\bf x}$'', as
$$P(i,i+n-1) \wedge (\forall i', \ 0 \leq i' < i, \ \neg P(i',i'+n-1) ).$$

From Lemma~\ref{lnpow} we know that the novel occurrences of
length-$n$ powers are clustered into a finite number of blocks.   
Then, as in the proof of Theorem~\ref{scthm}, we can guess the
endpoints of these blocks, and verify that the length-$n$ factors
beginning at the positions inside the blocks are novel occurrences
of powers, while those
outside are not, and sum the lengths of the blocks,
using a finite automaton built from $M$.  Thus, the function counting
the number of length-$n$ powers in $\bf x$ is $k$-synchronized.

The number of length-$n$ primitive words in $\bf x$ is then also
$k$-synchronized, since it
is expressible as the total number of
words of length $n$, minus the number of length-$n$ powers.
\end{proof}

\begin{remark}
Using the technique above, we can prove analogous results for the
functions counting the number of length-$n$ words that are
$\alpha$-powers, for any fixed rational number $\alpha> 1$.
\end{remark}

\section{Unsynchronized sequences}

It is natural to wonder whether other aspects of $k$-automatic sequences
are always
$k$-synchronized.  We give an example that is not.

We say a word $w$ is {\it bordered} if it has a nonempty prefix,
other than $w$ itself, that is also a suffix.  Alternatively, 
$w$ is bordered if it can be written in the form $w = tvt$, where
$t$ is nonempty.    Otherwise a word is {\it unbordered}.

Charlier et al.\ \cite{Charlier&Rampersad&Shallit:2011} showed that 
$u_{\bf x} (n)$, the
number of unbordered factors of length $n$ 
of a sequence $\bf x$, is $k$-regular if $\bf x$ is $k$-automatic.
They also gave a conjecture for 
recursion relations defining $u_{\bf t} (n)$ where
$\bf t$ is the Thue-Morse sequence; this conjecture has
recently been verified by Go\v{c} and Shallit.

We give here an example of a $k$-automatic sequence where the number
of unbordered factors of length $n$ is not $k$-synchronized.

Consider the characteristic sequence of the powers of $2$:
$$ {\bf c} := 011010001000000010 \cdots .$$

\begin{theorem}
The sequence $\bf c$ is $2$-automatic, but the function $u_{\bf c}(n)$ counting
the number of unbordered factors is not $2$-synchronized.
\end{theorem}

\begin{proof}
It is not hard to verify that $\bf c$ is $2$-automatic and
that $\bf c$ has exactly $r+2$ unbordered
factors of length $2^r + 1$, for $r \geq 2$ --- namely, the
factors beginning at positions $2^i$ for $0 \leq i \leq r$,
and the factor beginning at position $2^r + 1$.
However, if $u_{\bf c}(n)$ were $2$-synchronized, then reading an input
where the first component looks like $0^i 1 0^r 1$ (and hence
a representation of $2^{r-1} + 1$) for large $r$ would force
the transitions to enter a cycle.  If the corresponding transitions
for the second component contained a nonzero entry, this would
force $u_{\bf c}(n)$ to grow linearly with $n$ when $n$ is of the form
$2^r + 1$.  Otherwise, the corresponding transitions for the second
component are just $0$'s, in which case $u_{\bf c}(n)$ is bounded above
by a constant, for $n$ of the form $2^r + 1$.  Both cases lead to
a contradiction.
\end{proof}

\end{document}